\documentclass{CSML}
\pdfoutput=1

\usepackage{lastpage}
\newcommand{\dOi}{14(2:14)2018}
\lmcsheading{}{1--\pageref{LastPage}}{}{}%
{Mar.~28, 2017}{Jun.~19, 2018}{}

\usepackage{amsmath,amssymb,amsfonts,latexsym,stmaryrd,diagrams}

\usepackage{hyperref}
\hypersetup{hidelinks}

\newcommand{\Kc}{\mathcal{K}}
\newcommand{\KV}{\mathcal{KV}}
\newcommand{\Mod}{{\mathbf{Mod}}} 
\newcommand{\Asm}{{\mathbf{Asm}}}
\newcommand{\RT}{{\mathbf{RT}}}
\newcommand{\QCB}{{\mathbf{QCB}}}
\newcommand{\TP}{{\mathbf{TP}}}
\newcommand{\TD}{{\mathbf{TD}}}

\newcommand{\NN}{{\mathbb{N}}}
\newcommand{\RR}{{\mathbb{R}}}

\newcommand{\forget}[1]{}
\newcommand{\HS}{{\mathcal{H}}}
\newcommand{\HL}{{\mathcal{L}}}

\newcommand{\Obs}{\mathsf{Obs}}
\newcommand{\Sta}{{\mathsf{St}}}
\newcommand{\Ic}{{\mathcal{I}}}
\newcommand{\Ib}{{\mathbb{I}}}
\newcommand{\scpr}[2]{\langle #1 | #2 \rangle}

\newcommand{\tr}{{\mathsf{tr}}}
\newcommand{\AdmRep}{{\mathbf{AdmRep}}}
\newcommand{\Bfrak}{{\mathfrak{B}}}
\newcommand{\Rat}{{\mathbb{Q}}}
\newcommand{\Real}{{\mathbb{R}}}
\newcommand{\Complex}{{\mathbb{C}}}

\newcommand{\Bc}{\mathcal{B}}

\DeclareMathOperator{\closed}{{\mathcal{C}}}

\DeclareMathOperator{\Prj}{\mathsf{Prj}}

\newcommand{\CC}{{\mathbb{C}}}
\newcommand{\BB}{{\mathbb{B}}}

\usepackage{cite}

\begin{document}

\title{Computability in Basic Quantum Mechanics}
\dedicatory{Dedicated to Ji\v r\' i Ad\' amek on the occasion of his retirement}

\author[E.~Neumann]{Eike NEUMANN}
\address{Aston University, School of Engineering and Applied Science,
         Aston Triangle, Birmingham B4 7ET, UK}
\email{neumaef1@aston.ac.uk}
\author[M.~Pape]{Martin Pape}
\address{Kallinchener Str.~2B, Berlin D-15749, Germany}  
\email{martin.pape@gmx.de}
\author[T.~Streicher]{Thomas STREICHER}	
\address{Fachbereich 4 Mathematik TU Darmstadt,
         Schlo{\ss}gartenstr.~7, D-64289, Germany}	
\email{streicher@mathematik.tu-darmstadt.de}  

\begin{abstract}
The basic notions of quantum mechanics are formulated in terms of
separable infinite dimensional Hilbert space $\HS$. In terms of the Hilbert
lattice $\HL$ of closed linear subspaces of $\HS$ the notions of \emph{state} 
and \emph{observable} can be formulated as kinds of measures as in 
\cite{PtakPulmannov1991:OrthomodularStructuresAsQuantumLogic}. The aim of 
this paper is to show that there is a good notion of computability for these 
data structures in the sense of Weihrauch's Type Two Effectivity (TTE)
\cite{Weihrauch2000:ComputableAnalysis}.

Instead of explicitly exhibiting admissible representations for the data 
types under consideration we show that they do live within the category 
$\QCB_0$ which is equivalent to the category $\AdmRep$ of admissible 
representations and continuously realizable maps between them. For this purpose
in case of observables we have to replace measures by \emph{valuations} which 
allows us to prove an effective version of von Neumann's Spectral Theorem.
\end{abstract}

\maketitle

\section{Introduction}

In his legendary book \cite{Neumann1932:MathematischeGrundlagenDerQuantenmechanik} 
from 1932 J.~von Neumann gave a mathematical formulation of basic quantum mechanics 
based on separable Hilbert space $\HS$ which may manifest itself as $\ell^2$ as in
Heisenberg's \emph{matrix mechanics} or $L^2$ as in Schr\"odinger's \emph{wave mechanics}. 
In this setting \emph{observables} appear as \emph{self adjoint operators} on $\HS$ 
and \emph{states} as particular observables, namely so-called \emph{density operators} 
which are self adjoint operators $D \geq 0$ with $\tr(D) = 1$. The latter are closed
under countable convex combinations. Those states which cannot be obtained as non-trivial
countable convex combinations are called \emph{pure} and correspond to 1-dimensional 
subspaces of $\HS$.

It is \emph{a priori} not clear why observables should be understood as self adjoint
operators on Hilbert space. But this mystery is explained by von Neumann's famous
\emph{Spectral Theorem} already proved 
in \cite{Neumann1932:MathematischeGrundlagenDerQuantenmechanik} which establishes a 
1-1-correspondence between self adjoint operators on $\HS$ and
\emph{projector valued measures} on $\HS$, i.e.\ measures on $\RR$ taking values 
not in the unit interval $\Ib = [0,1]$ but in the so-called \emph{Hilbert lattice} 
$\HL$ of closed linear subspaces of $\HS$ which classically correspond to projectors, 
i.e.\ self adjoint operators $P = P^2$ on $\HS$. A projector valued measure 
$o : \Bfrak(\RR)\to\HL$ together with a pure state as given by a unit vectors $x$ in $\HS$ 
gives rise to an ordinary probability measure $s_x \circ o$ where $s_x(P) = \scpr{x}{Px}$.  
This is explained in detail in the book
\cite{PtakPulmannov1991:OrthomodularStructuresAsQuantumLogic} where it is also shown 
that states may be understood as measures on the Hilbert lattice $\HL$. More details 
will be given subsequently in subsections \ref{qstat} and \ref{qobs}, respectively.

In most physics textbooks one finds only the functional analytical account where 
observables are self adjoint operators because it is based on more traditional 
mathematics and more useful for (symbolic) computation (done by hand). There is a vast 
literature on the so-called ``logico-algebraic'' account based on the Hilbert lattice $\HL$. 
It goes back to old work of Birkhoff and von Neumann where they proposed to consider $\HL$ 
as a kind of ``quantum logic''. But since the lattice $\HL$ is not distributive it interprets
neither intuitionistic nor classical logic. In our paper we will not misuse $\HL$ for logical
purposes but rather as a tool for presenting a more algebraic and conceptual account of
basic quantum mechanics as in \cite{PtakPulmannov1991:OrthomodularStructuresAsQuantumLogic}
(despite its title).

The functional analytic formulation has already been studied in the framework
of Type Two Effectivity 
\cite{BrattkaDillhage2005:ComputabilityOfTheSpectrumOfSelfAdjointOperators,
WeihrauchZhong2006:ComputingSchroedingerPropagatorsOnType2TuringMaschines}.
To our knowledge computability for the ``logico-algebraic'' approach has not 
been considered so far in the literature.
At first sight this seems to be impossible since constructively they are not equivalent
because projectors on separable infinite dimensional Hilbert space correspond to
\emph{located} closed linear subspaces and these are not even closed under binary
intersection, see \cite{Dediu2000:ConstructiveTheoryOfOperatorAlgebras}. Nevertheless, 
we will show that the basic notions of the ``logico-algebraic approach'' can be endowed 
with an appropriate notion of computability in the sense of 
\cite{Weihrauch2000:ComputableAnalysis}. The key idea is to identify $\HL$ not as a
$\neg\neg$-subobject of $\Bfrak(\HS)$ but as a $\neg\neg$-subobject of $\Sigma^{\HS^\prime}$
where $\Sigma$ is the Sierpi\'nski space and $\HS^\prime$ is the dual space of $\HS$.
Classically, the space $\HS^\prime$ is anti-isomorphic to $\HS$ which, however, is not
the case computationally. Rather it turns out that the topology on $\HS^\prime$ induced
by computability is the sequentialization of the weak$^*$ topology.

In K.~Weihrauch's book \cite{Weihrauch2000:ComputableAnalysis} one finds a 
theory of computability for classical spaces based on Turing machines with 
infinite input and output tapes. Based on \cite{Troelstra1992:ComparingTheTheoryOfRepresentationsAndConstructiveMathematics},
Bauer and Lietz have shown in \cite{Lietz2004:FromConstructiveMathematicsToComputableAnalysisViaTheRealizabilityInterpretation,Bauer2000:RealizabilityApproachToComputableAnalysis,Bauer2005:RealizabilityAsConnectionBetweenConstructiveAndComputableMathematics} 
that computable analysis can be rephrased in terms of constructive analysis inside 
the \emph{function realizability topos} $\RT(\Kc_2)$ or rather its restriction 
to effective morphisms, the so called \emph{Kleene-Vesley topos} $\KV$, 
as described in \cite{Oosten2008:Realizability}.
In \cite{Battenfeld2007:ConvenientCategoryOfDomains} it has been shown how
to characterize abstractly within $\RT(\Kc_2)$ the category $\AdmRep$ of 
\emph{admissible representations of spaces and continuous(ly realizable) maps} 
between them which forms the backbone of K.~Weihrauch's account in 
\cite{Weihrauch2000:ComputableAnalysis}. Analogously, by restricting to $\KV$ 
one obtains the category $\AdmRep_{\mathsf{eff}}$ of \emph{admissible 
representations of spaces and effectively realizable maps} between them
since effectively realizable is equivalent to the existence of a Turing
machine with infinite tapes performing the respective transformation of
infinite sequences of natural numbers.

As shown in \cite{Battenfeld2007:ConvenientCategoryOfDomains} the category 
$\AdmRep$ is equivalent to a (fairly) small full subcategory $\QCB_0$ of the 
category $\mathbf{Sp}$ of topological spaces and continuous maps, namely the
one on $T_0$ quotients of countably based $T_0$ spaces. The equivalence of $\QCB_0$
and $\AdmRep$ is essentially due to the fact that all countably based $T_0$ spaces
appear as quotients of subspaces of Baire space whose elements are used in 
\cite{Weihrauch2000:ComputableAnalysis} for representing elements of more 
abstract spaces. 

This category $\QCB_0$ and thus also $\AdmRep$ has excellent categorical closure
properties. In particular, it is \emph{cartesian closed} and closed under 
regular, i.e.\ classical, subobjects. Within $\AdmRep \simeq \QCB_0$ one finds 
all complete separable metric spaces and, accordingly, it is a natural place 
for the Hilbert space approach to Quantum Mechanics as introduced by 
von Neumann in \cite{Neumann1932:MathematischeGrundlagenDerQuantenmechanik}.

In our account, however, we provide a notion of computability for the 
more algebraic approach based on the Hilbert lattice $\HL$ as described in
\cite{PtakPulmannov1991:OrthomodularStructuresAsQuantumLogic}. Due to the 
closure properties of $\AdmRep \simeq \QCB_0$ and the fact that it hosts the 
Sierpi\'nski space $\Sigma$ and Hilbert space $\HS$ it also hosts its dual 
$\HS^{\prime}$ and the classical subobject $\HL$ of $\Sigma^{\HS^\prime}$ on
closed subspaces of $\HS^\prime$. Notice that a closed linear subspace $P$ 
of $\HS^\prime$ is represented by the continuous map $p \in \Sigma^{\HS^\prime}$ 
with $P = p^{-1}(\bot)$, i.e.\ somewhat surprisingly $\bot \in \Sigma$ plays 
the role of ``true''. As a consequence the natural order induced by $\Sigma$ 
on $\HL$ is {\bf opposite} to subset inclusion as considered usually.

Let $\Ic$ be the unit interval $[0,1]$ with the \emph{upper topology}, i.e.\ the
Scott topology on the continuous lattice $([0,1],\geq)$. In $\AdmRep\simeq\QCB_0$ 
we will identify the space $\Sta$ of quantum {\bf states} as the $\neg\neg$-subobject 
of $\Ic^\HL$ consisting of those $s$ which validate the conditions
\begin{enumerate}
\item[(S1)] $s(0) = 0$ and $s(\HS) = 1$
\item[(S2)] $s(P \vee Q) = s(P) + s(Q)$ whenever $P \perp Q$
\end{enumerate}
since $s$ is continuous and thus preserves infima of decreasing $\omega$-chains.\footnote{As 
usual $P \perp Q$ stands for $\forall x \in P.\forall y \in Q.\,\scpr{x}{y} = 0$.}

By the spectral theorem for self-adjoint operators on $\HS$ quantum {\bf observables}
correspond to \emph{projection valued measures} on $\RR$, i.e.\ certain maps from 
the set $\Bfrak(\RR)$ of Borel subsets of $\RR$ to $\HL$. But since $\Bfrak(\RR)$ 
does not live within $\AdmRep \simeq \QCB_0$ we have to restrict to a generating 
subcollection. It turns out that the object $\closed(\RR)$ of closed subsets 
of $\RR$ is a good choice for this purpose since observables can be characterized 
as those $\nu \in \HL^{\closed(\RR)}$ for which the map
$\lambda C \in \closed(\RR). \scpr{x}{\nu(C)x}$ is a probability valuation on $\RR$
for all unit vectors $x \in \HS$.

Based on this reformulation of observables we will prove that the 
Spectral Theorem for bounded observables is effective in the sense that it holds 
in $\KV$. It will turn out that the induced topology on the operator side is the
sequentialization of the strong operator topology and that a sequence $(\nu_n)$ of
observables converges to $\nu_\infty$ w.r.t.\ the induced topology iff for all unit 
vectors the associated measures converge in the sense usually considered in the 
respective literature \cite{Billingsley99:ConvergenceOfProbabilityMeasures}.

\section{Basic Quantum Mechanics}\label{bqm}

We briefly recall how basic quantum mechanics can be formulated in terms of
separable infinite dimensional complex Hilbert space $\HS$ (see e.g.\
\cite{Mackey1963:MathematicalFoundationsOfQuantumMechanics,Prugovecki1971:QuantumMechanicsInHilbertSpace} for background information) as pioneered in J.~von Neumann's book from 1932 
\cite{Neumann1932:MathematischeGrundlagenDerQuantenmechanik}. 

\subsection{Basics Facts about Hilbert Space}

Up to isomorphism there is just one separable infinite dimensional Hilbert space $\HS$
over the field $\CC$ of complex numbers, namely the space $\ell^2$ of sequences $x$ of 
complex numbers such that $\sum |x_n|^2$ converges. Addition and scalar multiplication 
is pointwise and the scalar product is given by $\scpr{x}{y} = \sum x_n^* y_n$ with $x_n^*$ 
the complex conjugate of $x_n$. It is known to be a Banach space w.r.t.\ the norm 
$\|x\| = \sqrt{\scpr{x}{x}}$. There is a canonical countable orthonormal basis $(e_n)$
for $\ell^2$ where the $n$-th component of $e_n$ is $1$ and all other components are $0$.
We have $\scpr{e_n}{e_n} = 1$ and $\scpr{e_n}{e_m} = 0$ for $n \neq m$ and 
$x = \sum\limits_{n=0}^\infty \scpr{e_n}{x} e_n$ for all $x \in \HS$.

We recall for later use that the \emph{weak topology} on $\HS$ is the 
coarsest topology for which every linear functional of the form
$y \mapsto \scpr{x}{y}$ is continuous. The weak topology on $\HS$ is known 
to be Hausdorff and the unit ball $B(\HS) := \{x \in \HS \mid \|x\| \le 1\}$ 
is compact w.r.t.\ the weak topology but not w.r.t.\ the norm topology. Moreover,
subspaces of $\HS$ are closed w.r.t.\ the norm topology iff they are closed 
w.r.t.\ the weak topology on $\HS$. Notice, moreover, that $\HS$ with the weak 
topology is isomorphic to the dual space $\HS^\prime$ with the weak$^*$ topology, 
i.e.\ the coarsest topology rendering continuous all evaluation maps 
$\HS^\prime \to \Complex : f \mapsto f(x)$ for $x \in \HS$.

We write $\Bc(\HS)$ for the space of \emph{bounded linear operators} on $\HS$.
An $A \in \Bc(\HS)$ is called \emph{self-adjoint} iff 
$\scpr{Ax}{y} = \scpr{x}{Ay}$ for all $x,y \in \HS$. Such an $A$ is called
\emph{positive} iff $\scpr{x}{Ax} \geq 0$ for all $x \in \HS$ and it is called 
an \emph{effect} iff $0 \leq \scpr{x}{Ax} \leq 1$ for all $x \in \HS$ with
$\|x\| = 1$. A \emph{projector} is a self-adjoint $P \in \Bc(\HS)$ which, 
moreover, is idempotent, i.e.\ $PP = P$. As is well known projectors correspond 
to closed linear subspaces of $\HS$ where a projector $P$ maps $x \in \HS$ to the 
best approximating element $P(x)$ of the closed subspace of $\HS$ as given by 
fixpoints of $P$.

The trace $\tr(A)$ of a positive self-adjoint $A \in \Bc(\HS)$ is 
$\sum_n \scpr{e_n}{Ae_n}$ which exists iff this sum is bounded. 
Notice that $\tr(A)$ is independent from the choice of the orthonormal basis.
A positive self-adjoint operator with trace $1$ is called a \emph{density operator}.
Positive self-adjoint operators with trace $\leq 1$ are often called \emph{partial states}.

\subsection{Hilbert Lattice}

The \emph{Hilbert lattice} $\HL$ consists of the closed linear subspaces of $\HS$ 
ordered by subset inclusion. The poset $\HL$ is a lattice where meets are given
by intersections and joins are given by closures of linear spans of unions. 
The bottom element of $\HL$ is the zero subspace $0$ of $\HS$ whereas the top
element of $\HL$ is $\HS$. Classically, we may identify a closed linear subspace 
of $\HS$ with the corresponding projector of $\HS$ on this subspace. 

Notably, the Hilbert lattice $\HL$ is \emph{not} distributive and thus neither
boolean nor a complete Heyting algebra. Nevertheless, for every $P \in \HL$ we may 
consider its \emph{orthocomplement} 
\[ P^\perp = \{x \in \HS \mid \forall y \in P.\, \scpr{x}{y} = 0\} \]
which again is an element of $\HL$. 
Notice that orthocomplementation $(\cdot)^\perp : \HL \to \HL$ reverses the order, 
i.e.\ $Q^\perp \subseteq P^\perp$ whenever $P \subseteq Q$, and is involutory in the 
sense that $P^{\perp\perp} = P$. We write $P \perp Q$ for $P \subseteq Q^\perp$ stating 
that all vectors in $P$ are orthogonal to all vectors in $Q$. Notice that we always 
have $P \vee P^\perp = \HS$ and $P \wedge P^\perp = 0$ though typically there will be 
many different $Q \in \HL$ with $P \vee Q = \HS$ and $P \wedge Q = 0$ in contrast to 
boolean algebras where such complements are unique. 
A distinguishing property of $\HL$ is the law of \emph{orthomodularity} stating 
that $$Q = P \vee (Q \wedge P^\perp)$$ for $P \subseteq Q$.

It is well known, see e.g.\ \cite{Halmos1967}, that subspaces of $\HS$ are closed 
w.r.t.\ the norm topology iff they are closed w.r.t.\ the weak topology. Thus, we may 
identify $\HL$ with closed linear subspaces of $\HS^\prime$ endowed with the 
weak$^*$ topology which is homeomorphic to $\HS$ with the weak topology.

\subsection{Quantum States}\label{qstat}

We recall the basic notions of quantum mechanics as can be found in the classical 
text \cite{Mackey1963:MathematicalFoundationsOfQuantumMechanics} though we essentially 
follow  the equivalent presentation of \cite{PtakPulmannov1991:OrthomodularStructuresAsQuantumLogic}.

A (quantum) \emph{state} is a function $s$ from $\HL$ to the unit interval
$\Ib = [0,1]$ satisfying the conditions
\begin{enumerate}
  \item[(s1)] $s(0) = 0$ and $s(\HS) = 1$
  \item[(s2)] $s(\bigvee_n P_n) = \sum_n s(P_n)$ whenever $P_n \perp P_m$ for $n\neq m$.
\end{enumerate}
Thus, one may think of a state as a kind of probability measure on $\HL$ 
where disjointness is replaced by orthogonality. 

Every density operator $D$ on $\HS$ induces a state
\[  s(P) = \tr(DP) \]
where $P$ on the right hand side refers to the corresponding projector on $\HS$. 
By the famous \emph{Gleason's Theorem} (see e.g.\ \cite{Gleason1957:MeasuresOnClosedSubspaces,PtakPulmannov1991:OrthomodularStructuresAsQuantumLogic}) this establishes a 1-1-correspondence between states and density operators. 

A state $s$ on $\HL$ is called \emph{pure} iff there is a unit vector 
$x \in \HS$ with
\[ s(P) = s_x(P) = \scpr{x}{Px} \]
for all $P \in \HL$. One can show that

\begin{prop}\label{convstat}
Every state $s$ can be written as 
$\sum\limits_{n=0}^\infty \lambda_n s_{b_n}$ where the $\lambda_n \in \Ib$ with 
$\sum \lambda_n = 1$ and $(b_n)$ is an orthonormal basis for $\HS$.
\end{prop}

\subsection{Quantum Observables}\label{qobs}

A (quantum) \emph{observable} is a function $o$ from the set $\Bfrak(\RR)$ 
of Borel subsets of $\RR$ to $\HL$ such that 
\begin{enumerate}
  \item[(o1)] $o(\emptyset) = 0$ and $o(\RR) = \HS$
  \item[(o2)] $o(\bigcup_n X_n) = \bigvee_n o(X_n)$ 
  \item[(o3)] $o(X) \perp o(Y)$ whenever $X \cap Y = \emptyset$
\end{enumerate}
i.e.\ $o$ is a \emph{projector valued measure} (when identifying elements
of $\HL$ with projectors).

By the famous von Neumann Spectral Theorem \cite{Neumann1932:MathematischeGrundlagenDerQuantenmechanik,Prugovecki1971:QuantumMechanicsInHilbertSpace} bounded self-adjoint operators
$A$ on $\HS$ correspond to observables $o$ which are bounded in the sense that
$o([x,y]) = \HS$ for some $x \leq y$ in $\RR$ via
$$\scpr{x}{Ay} = \int_\RR \lambda \; d\scpr{x}{o((-\infty,\lambda)) y}$$
making use of the fact that $X \mapsto \scpr{x}{o(X)(y)}$ is a $\CC$-valued
measure on $\RR$ (see e.g.\ \cite{Prugovecki1971:QuantumMechanicsInHilbertSpace}.

Notice that an observable $o : \Bfrak(\RR) \to \HL$ composed with state 
$s : \HL \to \Ib$ gives rise to a probability measure 
$s \circ o : \Bfrak(\RR) \to \Ib$ on $\RR$.
 
\subsection{Alternative Characterizations of States and Observables}

For later use in section~\ref{cbqm} it is useful to consider the 
following alternative characterizations of states and observables.

\begin{prop}\label{statchar}
A map $s : \HL \to \Ib$ is a state iff it satisfies the conditions
\begin{enumerate}
\item[(S1)] $s(0) = 0$ and $s(\HS) = 1$
\item[(S2)] $s$ preserves infima of decreasing $\omega$-chains
\item[(S3)] $s(P \vee Q) = s(P) + s(Q)$ whenever $P \perp Q$.
\end{enumerate}
\end{prop}

\begin{proof}
It is well known that states validate the condition (S1)--(S3) as shown e.g.\
in \cite{PtakPulmannov1991:OrthomodularStructuresAsQuantumLogic}.

For the reverse direction suppose $s : \HL \to \Ib$ validates conditions 
(S1)--(S3). Condition (s1) holds since it is the same as (S1).  From this
together with (S3) and $P \vee P^\perp = \HS$ for all $P \in \HL$, 
it is immediate that $s(P^\perp) = 1 - s(P)$. For this reason from (S2) it follows 
that $s$ preserves suprema of increasing $\omega$-chains. For showing $(s2)$ suppose 
$(P_k)$ is a pairwise orthogonal sequence in $\HL$. Let $Q_n = \bigvee_{k=0}^n P_k$. 
Then $(Q_n)$ is an increasing $\omega$-chain in $\HL$ with $\bigvee P_k = \bigvee Q_n$.
Thus, we have
\[
\begin{array}{rlr} 
s(\bigvee_k P_k) & = s(\bigvee_n Q_n) &  \\
                 & = \sup_n s(Q_n) &  \\
                 & = \sup_n s(\bigvee_{k=0}^n P_k) &  \mbox{(S3)} \\
                 & = \sup_n \sum_{k=0}^n s(P_k) &  \\
                 & = \sum_k s(P_k) & 
\end{array}
\]
showing that $s$ satisfies (s2) and thus is a state.
\end{proof}

For well behaved spaces $X$ like $\RR$ probability measures on $X$ are uniquely 
determined by their restrictions to $\closed(X)$, the set of 
closed subsets of $X$, which forms a(n $\omega$-)cpo w.r.t.\ $\supseteq$, 
see \cite{Gierz2003:ContinuousLattices}. These restrictions of probability measures 
on $X$ can be characterized as \emph{valuations}, i.e.\ Scott continuous maps $\nu$ 
from $\closed(X)$ to $\Ic$, the unit interval $\Ib$ ordered by $\geq$, satisfying
$\nu(\emptyset) = 0$, $\nu(\RR) = 1$ and
\[ \nu(A) + \nu(B) = \nu(A \cup B) + \nu(A \cap B) \]
for $A,B \in \closed(X)$. See \cite{Gierz2003:ContinuousLattices} 
for more information on valuations though formulated there in terms of open 
instead of closed subsets of $X$.

We will now characterize quantum observables in terms of valuations.

\begin{prop}\label{obschar}
Quantum observables correspond by restriction to $\closed(\RR)$ to
\emph{quantum valuations}, i.e.\ maps $\nu : \closed(\RR) \to \HL$
such that
\begin{enumerate}
  \item[(O1)] $\nu(\emptyset) = 0$ and $\nu(\RR) = \HS$
  \item[(O2)] $\nu$ preserves infima of decreasing $\omega$-chains
  \item[(O3)] for every unit vector $x$ in $\HS$ and $A,B \in \closed(\RR)$
                     \[\nu_x(A) + \nu_x(B) = \nu_x(A \cup B) + \nu_x(A \cap B) \]
                     where $\nu_x(C) = \scpr{x}{\nu(C)(x)}$ for $C \in \closed(\RR)$.
\end{enumerate}
\end{prop}

\begin{proof}
Let $o : \Bfrak(\RR) \to \HL$ be an observable and $\nu$ its restriction to
$\closed(\RR)$. Condition (O1) for $\nu$ is immediate from condition (o1) 
for $o$. As shown in \cite{PtakPulmannov1991:OrthomodularStructuresAsQuantumLogic}
$o$ preserves infima of decreasing $\omega$-chains from which (O2) is immediate.
For every unit vector $x$ in $\HS$ the function $o_x = s_x \circ o$ is a measure
from which (O3) is immediate.

Suppose $\nu : \closed(\RR) \to \HL$ validates conditions (O1)-(O3). Then
$E : \RR \to \HL : \lambda \mapsto 1 - \nu([\lambda,\infty))$ is a spectral
family in the sense of \cite{Prugovecki1971:QuantumMechanicsInHilbertSpace}
which as shown in \emph{loc.cit.}\ uniquely extends to an observable $o :
\Bfrak(\RR) \to \HL$ with $o((-\infty,\lambda)) = E(\lambda)$.
\end{proof}

\section{Topological Domain Theory}\label{topdom}

The Hilbert lattice $\HL$ is complete and thus in particular a directed complete 
poset as studied in denotational semantics (see e.g.\ \cite{Gierz2003:ContinuousLattices,Streicher2006:DomainTheoreticFoundationsOfFunctionalProgramming}). However, we want to 
arrive at a notion of computability for the Hilbert lattice and derived notions such as 
states and observables and this is not possible for arbitrary directed complete posets or 
complete lattices.
The first idea would be to exhibit $\HL$ as an \emph{effectively given domain}
as described e.g.\ in \cite{Streicher2006:DomainTheoreticFoundationsOfFunctionalProgramming}. 
But for this purpose $\HL$ would have to be at least a \emph{continuous} lattice in the sense
of \cite{Gierz2003:ContinuousLattices} which, alas, is not the case as has been pointed out 
to us by K.~Keimel.

\begin{prop}
The Hilbert lattice $\HL$ is not continuous.
\end{prop}
\begin{proof}
Suppose $\HL$ were a continuous lattice. Then every atom $a$ of $\HL$ were 
compact. But there is an atom $a$ in $\HL$ such that for no $n$ we have
$a \le \bigvee_{i=0}^n \langle e_n \rangle$ (where $\langle e_n \rangle$ is the
one dimensional subspace spanned by $e_n$).
But $a \le \HS = \bigvee_{n=0}^\infty \langle e_n \rangle$ and thus $a$ is not compact.
\end{proof}

Due to this shortcoming we will instead work in the framework of 
\emph{topological domain theory} as described in \cite{Battenfeld2007:ConvenientCategoryOfDomains} which subsumes both countably based continuous domains and 
complete separable metric spaces. 

The basic idea of topological domain theory is to identify an appropriate full 
subcategory of the \emph{function realizability topos} $\RT(\Kc_2)$ 
(as described e.g.\ in \cite{Oosten2008:Realizability}) which is equivalent to 
the category $\AdmRep$ of \emph{admissible representations} and 
\emph{continuously realizable} maps between them.

\subsection{Admissible Representations}

Admissible representations are the basic structures underlying Weihrauch's
\emph{Type Two Effectivity} (TTE) as described in 
\cite{Weihrauch2000:ComputableAnalysis}. We briefly recall some basic notions.

The set of all functions from $\NN$ to $\NN$ endowed with the \emph{initial 
segment} topology is commonly called \emph{Baire space} for which we write
$\BB$. A \emph{representation} of a topological $T_0$ space $X$ is a 
\emph{quotient} map $\rho$ from a subspace $B$ of $\BB$ to $X$. 
For representations $\rho : B \to X$ and $\rho^\prime : B^\prime \to X^\prime$ 
a function $f : X \to X^\prime$ is called \emph{continuously realizable} iff 
there exists a continuous function $\phi : B \to B^\prime$ making the diagram
\begin{diagram}[small]
B & \rTo^\phi & B^\prime \\
\dTo^\rho & & \dTo_{\rho^\prime} \\
X & \rTo_f& X^\prime
\end{diagram}
commute. A representation $\rho : B \to X$ is called \emph{admissible} iff
for every continuous map $f$ from a subspace $B^\prime$ of $\BB$ to $X$ there
is a continuous map $\phi : B^\prime \to B$ rendering the triangle
\begin{diagram}[small]
B^\prime & \rTo^\phi & B \\
& \rdTo_f & \dTo_\rho \\
& & X
\end{diagram}
commutative. It is easy to see that for admissible representations $\rho : B \to X$
and $\rho^\prime : B^\prime \to X^\prime$ a map $f : X \to X^\prime$ is continuous
iff it is continuously realizable as a map from $\rho$ to $\rho^\prime$. We
write $\AdmRep$ for the ensuing category of admissible representations and
continuous(ly realizable) maps between them. 

We recall from \cite{Battenfeld2007:ConvenientCategoryOfDomains} that complete
separable metric spaces and countably based continuous domains form full
subcategories of $\AdmRep$. 

\subsection{\texorpdfstring{$\QCB_0$}{QCB\_0} Spaces}

As discussed in \cite{Battenfeld2007:ConvenientCategoryOfDomains} the category 
$\AdmRep$ is equivalent to the following subcategory of the category of 
topological spaces and continuous maps. 

\begin{defi}
A \emph{$\QCB_0$ space} is a $T_0$-quotient of a countably based topological space.
We write $\QCB_0$ for the ensuing category of $\QCB_0$ spaces and 
continuous maps between them.
\end{defi}

The $\QCB_0$ spaces are precisely those topological $T_0$ spaces which admit
an admissible representation. Moreover, as shown in \cite{Battenfeld2007:ConvenientCategoryOfDomains}(4.10) the category $\QCB_0$ is cartesian closed, countably complete 
and countably cocomplete. 

Let $\Sigma$ be the Sierpi\'nski space $\{\bot,\top\}$ whose only nontrivial
open set is $\{\top\}$. Obviously, continuous maps from $X$ to $\Sigma$
correspond to open subsets of $X$. For further reference we recall the 
following useful fact from \cite{Battenfeld2008:TopologicalDomainTheory}.

\begin{prop}\label{powerSigSc}
For every $\QCB_0$ space $X$ the exponential $\Sigma^X$ in $\QCB_0$ is 
isomorphic to the space $\mathcal{O}(X)$ of open subsets of $X$ endowed with
the Scott topology arising from the subset ordering $\subseteq$.
\end{prop}
\begin{proof}
First of all the elements of $\Sigma^X$ are the open subsets of $X$ and the
information ordering corresponds to $\subseteq$. Thus it suffices to show
that every Scott closed subset $C$ of $\Sigma^X$ is closed w.r.t.\ the 
topology of $\Sigma^X$.

For this purpose we recall the following result. Let $N_\infty$ be the one
point compactification of $\NN$ which is countably based and thus in $\QCB_0$.
One can show that the map $\bigsqcap : \Sigma^{N_\infty} \to \Sigma : p
\mapsto \bigwedge_{n\in\NN} p_n$ is continuous.

Now suppose $p : N_\infty \to \Sigma^X$ with $p_n \in C$ for all $n\in\NN$.
Consider $q_n : X \to \Sigma : x \mapsto \bigwedge_{k\in\NN} p_{n+k}(x)$ which is
continuous since  $\bigsqcap : \Sigma^{N_\infty} \to \Sigma$ is continuous.
Obviously, we have $q_n \sqsubseteq p_n$ and thus $q_n \in C$ and $p_\infty = 
\bigsqcup_{n\in\NN} q_n$. Thus $p_\infty \in C$ since $C$ is Scott closed.
\end{proof}

On every topological space $X$ we may consider the \emph{specialization order} 
$$x \sqsubseteq_X y \;\equiv\; \forall O \in \mathcal{O}(X).\; x \in O \implies y \in O$$
which allows one to define the following notions.

\begin{defi}
A \emph{topological predomain} is a $\QCB_0$ space $X$ where 
every ascending $\omega$-chain $(x_n)$ (w.r.t.\ $\sqsubseteq_X$) 
has a least upper bound  $x_\infty$. We write $\TP$ for the 
\emph{category of topological predomains} which is a full subcategory of $\QCB_0$.

A \emph{topological domain} is a topological predomain $X$ which has a 
least element $\bot_X$ w.r.t.\ $\sqsubseteq_X$. We write $\TD$ for the
ensuing \emph{category of topological domains}.
\end{defi}

One can show that
\begin{prop}
Every continuous function between topological predomains 
preserves suprema of ascending $\omega$-chains.
\end{prop}

Furthermore as shown in \cite{Battenfeld2007:ConvenientCategoryOfDomains} 
it holds that

\begin{prop}
The category $\TP$ is a full reflective exponential ideal of $\QCB_0$.
\end{prop}

\begin{prop}
The category $\TD$ is an exponential ideal of $\QCB_0$ and is closed under
countable products in $\QCB_0$.
\end{prop}

\subsection{\texorpdfstring{$\AdmRep$}{AdmRep} within \texorpdfstring{$\RT(\Kc_2)$}{RT(K\_2)}}

Another important aspect of $\AdmRep$ is that it appears as a full reflective
subcategory of the function realizability topos $\RT(\Kc_2)$ as described e.g.
in \cite{Oosten2008:Realizability}. 

The underlying set of the ``second Kleene algebra'' $\Kc_2$ is Baire space. 
For $\alpha,\beta \in \BB$ we define $\alpha | \beta \simeq n$ iff 
$\alpha(\bar{\beta}(k)) = n+1$ and $\alpha(\bar{\beta}(\ell)) = 0$ 
for $\ell < k$.\footnote{We write $\bar{\alpha}(n)$
for the code of the sequence $\langle \alpha_0,\dots,\alpha_{n-1}\rangle$.}
The partial application operation of $\Kc_2$ is defined as
\[ \alpha\beta \simeq \gamma \iff 
   \forall n\in\NN.\; \alpha|(\langle n \rangle{*}\beta) = \gamma_n \]
where $*$ stands for concatenation of finite sequences with arbitrary
sequences. 

See \cite{Oosten2008:Realizability} for the definition of the category
$\Asm(\Kc_2)$ of \emph{assemblies} which is equivalent to the full subcategory 
of $\RT(\Kc_2)$ on $\neg\neg$-separated objects and its full subcategory 
$\Mod(\Kc_2)$ of \emph{modest sets}. Recall that modest sets are quotients
of $N^N$ in $\RT(\Kc_2)$ w.r.t.\ $\neg\neg$-closed partial equivalence 
where $N$ is the natural numbers object of $\RT(\Kc_2)$.

\begin{defi}
An object $X$ in $\RT(\Kc_2)$ is called \emph{$\Sigma$-extensional} if 
the map 
$$\eta_X : X \to \Sigma^{\Sigma^X} : x \mapsto \lambda p . p(x)$$
is a regular, i.e.\ $\neg\neg$-closed, monomorphism. We write $\Mod_\Sigma(\Kc_2)$ 
for the full subcategory of $\RT(\Kc_2)$ on $\Sigma$-extensional objects.
\end{defi}

Theorem~6.1.9 of \cite{Battenfeld2008:TopologicalDomainTheory} guarantees that
\begin{prop}\label{ExtCl}\leavevmode
\begin{enumerate}
\item[(1)] 
The category $\Mod_\Sigma(\Kc_2)$ is equivalent to the category $\QCB_0$.
\item[(2)] $\Mod_\Sigma(\Kc_2)$ is an exponential ideal in $\RT(\Kc_2)$.
\item[(3)] Up to isomorphism the objects of $\Mod_\Sigma(\Kc_2)$ are
           the $\neg\neg$-subobjects of powers of $\Sigma$.
\end{enumerate}
\end{prop}

\section{Computable Basic Quantum Mechanics}\label{cbqm}

The aim of this main section is to identify the Hilbert lattice $\HL$,
the type $\Sta$ of quantum states and the type $\Obs$ of quantum observables
as objects of $\AdmRep \simeq \QCB_0$. This will induce a notion of 
computability on $\HL$, $\Sta$ and $\Obs$ and suggest topologies on the 
respective sets which to our knowledge have not been considered so far in the 
literature on mathematical foundations of basic quantum mechanics. Of course, 
the sets $\HL$, $\Sta$ and $\Obs$ can all be identified with particular subsets
of $\Bfrak(\HS)$ which itself can be endowed with the various different topologies
as considered in (linear) functional analysis. We will discuss how these 
topologies relate to the ones induced by admissible representations.

\subsection{Separable Bananch Spaces within \texorpdfstring{$\AdmRep$}{AdmRep}}

As is well known from e.g.\ \cite{Weihrauch2000:ComputableAnalysis} all
complete separable metric spaces can be endowed with admissible 
representations. This applies in particular to $\RR$, $\CC$ and 
separable Banach spaces over these fields such as $\HS$.

\subsection{Spaces of Bounded Linear Operators within \texorpdfstring{$\AdmRep$}{AdmReP}}

For separable Banach spaces $E$ and $F$ there arises the question what
is the natural topology on the set $\Bfrak(E,F)$ of bounded linear operators
from $E$ to $F$. The norm topology endows $\Bfrak(E,F)$ with the structure
of a Banach space which, however, in general is not separable. This holds
in particular for $\Bfrak(\HS)$, the space of bounded linear operators 
on $\HS$, which can be seen as follows. Consider the linear operator 
$T : \ell^\infty \to \Bfrak(\ell^2)$ sending $x \in \ell^\infty$ to the 
linear operator $T(x) : \HS \to \HS : (y_n)_{n\in\NN} \mapsto (x_n y_n)_{n\in\NN}$
which has the same norm as $x$. Since $\ell^\infty$ is not separable the
Banach space $\Bfrak(\HS)$ is not separable w.r.t.\ the norm topology.\footnote{We thank V.~Brattka for drawing our attention to this counterexample.} 

However, since $E$ and $F$ are $\QCB_0$ spaces we may consider their exponential 
$F^E$ in $\QCB_0$. Following the description of exponentials in $\QCB_0$ as given 
in e.g.\ \cite{Battenfeld2007:ConvenientCategoryOfDomains} the underlying set of $F^E$
is the set of all continuous functions from $E$ to $F$ where $(f_n)$ converges 
to $f_\infty$ in $F^E$ iff for all sequences $(x_n)$ in $E$ converging to $x_\infty$ 
the sequence $(f_n(x_n))$ converges to $f_\infty(x_\infty)$ in $F$. 
Since being linear is a $\neg\neg$-closed predicate on $F^E$ we consider 
$\Bfrak(E,F)$ as the corresponding $\neg\neg$-closed subobject of $F^E$ 
which again is a $\QCB_0$ space. As shown in the next theorem the $\QCB_0$ topology
on $\Bfrak(E,F)$ is the sequentialization of a ``traditional'' topology 
on $\Bfrak(E,F)$, namely the strong operator topology.

\begin{thm}\label{convLO}
A sequence $(T_n)$ converges to $T$ in $\Bfrak(E,F)$ w.r.t.\ its $\QCB_0$ topology
iff $(T_n)$ converges to $T$ in the \emph{strong operator topology}.

Thus, the $\QCB_0$ topology of  $\Bfrak(E,F)$ is the sequentialization of the
strong operator topology on $\Bfrak(E,F)$.
\end{thm}
\begin{proof}
The forward direction is obvious.

For the reverse direction suppose that $(T_n)$ converges to $T$ in the
strong operator topology, i.e.\ $\lim\limits_{n\to\infty} T_n x = Tx$ 
for all $x \in E$. Thus, for all $x \in E$ the set 
$\{ T_n x \mid n \in \NN\} \cup \{Tx\}$ is bounded from which 
it follows by the Banach-Steinhaus theorem that 
$\{ \|T_n\| \mid n\in\NN\} \cup \{\|T\|\}$
is bounded by some $c > 0$. For showing that $(T_n)$ converges to $T$ in $F^E$ 
suppose that $(x_n)$ converges to $x$ in $E$. We have
\[\begin{array}{rl}
\|Tx - T_n x_n\|  &  \leq \|Tx - T_n x\| + \| T_n x - T_n x_n\| \\
                 &  \leq \|Tx - T_n x\| + \|T_n\| \cdot \| x - x_n\| \\
                 &  \leq  \|Tx - T_n x\| + c \|x_n - x\|
\end{array}\]
for which reason $\lim\limits_{n\to\infty} \|Tx - T_n x_n\| = 0$ since
$\lim\limits_{n\to\infty} \|Tx - T_n x\| = 0$ and 
$\lim\limits_{n\to\infty}\|x_n - x\| = 0$. 
Thus $\lim\limits_{n\to\infty} T_n x_n = Tx$ as desired.
\end{proof}

If $E$ is separable Hilbert space $\HS$ (e.g.\ $\ell^2$) and $F$ is $\CC$ or $\HS$
then the strong operator topology on $\Bfrak(E,F)$ is  {\em not sequential} (see solution 
of Problem 21 on p.185 of Halmos's \emph{Hilbert Space Problem Book}~\cite{Halmos1967}) 
for which reason one has to take its sequentialization to obtain the natural topology 
of $\Bfrak(E,F)$ in $\QCB_0$.
Thus, in particular, the natural topology on $E^\prime = \Bfrak(E,\CC)$ is the 
sequentialization of the weak$^*$ topology on $E^\prime$. 

Accordingly, in the following we will consider $\HS^\prime$ as endowed 
with the sequentialization of the weak$^*$ topology. We write $i : \HS\to\HS^\prime$
for the map with $i(x)(y) = \scpr{x}{y}$ which is continuous but not a homeomorphism
unless $\HS$ is endowed with the sequentialization of the weak topology.

\subsection{Hilbert Lattice within \texorpdfstring{$\AdmRep$}{AdmRep}}

Obviously, the Sierpi\'nski space $\Sigma$ also lives within $\AdmRep$.
Thus, also $\Sigma^{\HS^\prime}$ lives within $\AdmRep$. From Proposition~\ref{powerSigSc} 
we know that $\Sigma^{\HS^\prime}$ carries the Scott topology. Thus, when identifying
$p \in \Sigma^{\HS^\prime}$ with the closed subset $p^{-1}(\bot)$ of $\HS^\prime$ the set 
$\closed(\HS^\prime)$ of closed subsets of $\HS^\prime$ gets endowed with the 
Scott topology induced by the partial order $\supseteq$ for which we write
$\sqsubseteq_{\closed(\HS^\prime)}$ or simply $\sqsubseteq$ as is common for the
specialization order. 

However, for later use it is useful to make explicit what it means that
a sequence $(p_n)$ converges to $p_\infty$ in $\Sigma^{\HS^\prime}$, namely that
$(p_n(x_n))$ converges to $p_\infty(x_\infty)$ in $\Sigma$ whenever $(x_n)$ converges
to $x_\infty$ in $\HS^\prime$. Thus, the sequence $(p_n)$ converges to $p_\infty$ 
in $\Sigma^{\HS^\prime}$ iff for all $(x_n)$ converging to $x_\infty$ from 
$p_\infty(x_\infty) = \top$ it follows that $\exists n \forall k \geq n \; p_k(x_k) = \top$
iff for all $(x_n)$ converging to $x_\infty$ 
from $\forall n \exists k{\geq}n \; p_k(x_k) = \bot$ 
it follows that $p_\infty(x_\infty) = \bot$.

Since by Proposition~\ref{ExtCl} the category $\AdmRep$ is closed under 
$\neg\neg$-subobjects the collection of closed linear subspaces of $\HS^\prime$ gives 
rise to a $\neg\neg$-closed subobject of $\closed(\HS^\prime) \cong \Sigma^{\HS^\prime}$.

\begin{defi}\label{HLdef}
The \emph{Hilbert lattice} $\HL$ in $\RT(\Kc_2)$ is the subobject of
$\Sigma^{\HS^\prime}$ consisting of all $p$ satisfying the conditions
\begin{enumerate}
\item[(1)] $\forall x,y.\, p(x) = \bot \wedge p(y) = \bot \implies p(x+y) = \bot$
\item[(2)] $\forall x. \forall \lambda.\, p(x) = \bot \implies p(\lambda x) = \bot$.
\end{enumerate}
\end{defi}

Since conditions (1) and (2) are $\neg\neg$-closed $\HL$ appears as a
$\neg\neg$-subobject of $\Sigma^{\HS^\prime}$ and thus is an element of $\AdmRep$. 
As follows from \cite{Battenfeld2007:ConvenientCategoryOfDomains,Battenfeld2008:TopologicalDomainTheory} the topology on $\HL$ is the \emph{sequentialization} of the subspace 
topology induced by the inclusion of $\HL$ into $\Sigma^{\HS^\prime}$ which itself 
carries the Scott topology.

Nevertheless $\HL$ inherits its specialization order from $\Sigma^{\HS^\prime}$ 
as follows from

\begin{lem}\label{specordlm}
Let $A \subseteq_{\neg\neg} \Sigma^X$ in $\Mod(\Kc_2)$. Then $A$ inherits its information
ordering from $\Sigma^X$, i.e.\ for $p,q \in A$ we have $p \sqsubseteq q$ iff 
$p(x) \sqsubseteq q(x)$ for all $x \in X$. 
\end{lem}
\begin{proof}
By Proposition~\ref{powerSigSc} the claim holds for $\Sigma^X$. But for $p,q \in\Sigma^X$ 
we have $p \sqsubseteq q$ iff there exists a morphism $f : \Sigma \to \Sigma^X$ 
in $\Mod(\Kc_2)$ with $f(\bot) = p$ and $f(\top) = q$. 

Suppose $p,q \in A$. If $p \sqsubseteq_A q$ then $p \sqsubseteq_{\Sigma^X} q$ since
the inclusion of $A$ into $\Sigma^X$ is a morphism in $\Mod(\Kc_2)$. On the other hand 
if $p \sqsubseteq_{\Sigma^X} q$ then by the observation above there is a morphism 
$f : \Sigma \to \Sigma^X$ in $\Mod(\Kc_2)$ with $f(\bot) = p$ and $f(\top) = q$.
Since $f$ factors through $A$ it follows that $p \sqsubseteq_A q$.
\end{proof}

Now we can show that

\begin{prop}\label{HLdom}
$\HL$ is a topological domain.
\end{prop}

\begin{proof}
That $\HL$ is a topological predomain is immediate from the fact that 
it appears as an equalizer of maps between topological predomains 
corresponding to conditions (1) and (2) of Def.~\ref{HLdef}. 

Since the specialization order on $\HL$ is inherited from $\Sigma^{\HS^\prime}$ 
the least element of $\HL$ is given by the map $\bot_\HL : \HS^\prime \to \Sigma$ 
with $\bot_\HL(x) = \bot$ iff $x = 0$.
\end{proof}

Let $\Prj$ be the $\neg\neg$-subobject of $\Bfrak(\HS)$ consisting of projectors. 
Thus $\Prj$ is an object of $\AdmRep$ inheriting convergence from 
$\Bfrak(\HS) \subseteq_{\neg\neg}\HS^\HS$. Classically, every 
$p \in \HL \subseteq \Sigma^{\HS^\prime}$ can be identified with the corresponding 
projector $P_p \in \Bfrak(\HS)$ tacitly using that $\HS^\prime$ with the weak$^*$ topology
is homeomorphic to $\HS$ with the weak topology.
The bijective map from $\Prj$ to $\HL$ sending $P$ to $\{ x \in \HS \mid Px \neq x\}$ 
is continuous since definable in the internal language of $\RT(\Kc)$ but its inverse is not 
since it does not respect the specialization order. More generally, it holds that

\begin{prop}\label{diffHL}
The $\QCB_0$ spaces $\HL$ and $\Prj$ are not isomorphic.
\end{prop}
\begin{proof}
If $\HL$ and $\Prj$ were isomorphic then their specialization orders would be
isomorphic, too, which, however, is not the case since on $\Prj$ it is flat 
whereas $\HL$ has a least element w.r.t.\ its information ordering as follows 
from Lemma~\ref{specordlm}.
\end{proof}

Since $\RT(\Kc_2)$ is a model of Brouwerian intuitionistic mathematics (see 
\cite{KleeneVesley1965:TheFoundationsOfIntuitionisticMathematics,Oosten2008:Realizability})
it follows that one cannot prove constructively that $\HL$ and $\Prj$ are in 
1-1-correspondence.
Moreover, Prop.~\ref{diffHL} seems to show that von Neumann's Spectral Theorem does not 
hold in $\RT(\Kc_2)$ since from a classical point of view it entails a 1-1-correspondence
between closed linear subspaces of $\HS$ and projectors on $\HS$. But, as we will see 
later in subsection~\ref{BObsrep} this is not the case for an appropriate formulation 
of the Spectral Theorem since $\HL$ corresponds to spectral measures/valuations 
on $\Sigma$ rather than on the discrete space $2 = \{0,1\}$.

In view of Prop.~\ref{diffHL} the subsequent Th.~\ref{HLtowPr} might seem surprising. 
But in any case it will be crucial for proving our variant of the Spectral Theorem.
For formulating Th.~\ref{HLtowPr} we have to introduce a few conventions.

Let $S(\HS)$ be \emph{projective Hilbert space}, i.e.\ unit vectors of $\HS$ modulo
the equivalence relation $x \sim y \equiv \forall p \in \HL.\, p(x) = p(y)$. 
Notice that $x \sim y$ iff $x = \lambda y$ for some $\lambda \in \CC$ with 
$|\lambda| = 1$. 

Recall that $\Ic$ is the unit interval $[0,1]$ with the \emph{upper} topology
whose open sets are those downward closed subsets of $[0,1]$ which are open in the
usual Euclidean topology on $[0,1]$. Notice that $x \sqsubseteq_\Ic y$ iff $x \geq y$.
Moreover, one may characterize $\Ic$ as the Scott topology on the continuous lattice 
(in the sense of \cite{Gierz2003:ContinuousLattices}) $[0,1]$ ordered by $\geq$. For
this reason $\Ic$ may be called the \emph{upper interval}. We will write $\Ib_\leq$ for
the dual but isomorphic notion, namely $\Ib$ endowed with the \emph{lower} topology, 
i.e.\ the Scott topology induced by $\leq$ on $\Ib$.
 
\begin{thm}\label{HLtowPr}
The map $s : \HL \to \Ic^{S(\HS)} : P \mapsto x \mapsto \scpr{x}{Px}$ 
is a morphism in $\AdmRep$, i.e.\ continuous w.r.t.\ the induced topologies.
Moreover, the map $s$ is a $\neg\neg$-mono, i.e.\ $s$ is an iso when 
corestricted to its $\neg\neg$-image $\HL_p$. Moreover, both $s$ and
$s^{-1} : \HL_p \to \HL$ have effective realizers and thus live
in $\AdmRep_{\mathsf{eff}}$.
\end{thm}
\begin{proof} 
We give a constructive proof which is partly inspired by the proof of Th.~4.15 
in \cite{NeumannE2015}. The basic idea is to consider the isomorphic copy $\HL_q$ 
of $\HL_p$ induced by the isomorphism between $\Ic = \Ib_\geq$  and $\Ib_\leq$ 
sending $x$ to $\sqrt{1-x^2}$, i.e.\ $\HL_q$ consists of all functions 
from $S(\HS)$ to $\Ib_\leq$ which assign to $x \in S(\HS)$ the distance $d(x,L)$
to a closed linear subspace $L$ of $\HS^\prime$. We will use some notation and facts
as formulated and proven in Appendix~\ref{ENfacts}.

Let $L \in \HL$ be given as a $\Sigma$-predicate on $\HS^\prime$. Then 
$B(\HS) \cap L$ is also given by a $\Sigma$-predicate on $\HS^\prime$.
For $c \in S(\HS)$ and $0 \leq r < 1$ by Theorem~\ref{ENthm} we have 
$r < d(c,L)$ iff $B(\HS) \cap L \subseteq H_{c,r}$ which is in $\Sigma$
since $B(\HS) \cap L$ is a compact subset of $\HS^\prime$ and $H_{c,r}$
is open. Thus we have established the existence of $s$ as a map from $\HL$
to $\Ic_\leq^{S(\HS)}$ sending $L \in \HL$ to $s(L) = \lambda r.\, r < d(x,L)$.
Obviously $\HL_q$ is the $\neg\neg$-image of this $s : \HL \to \Ic_\leq^{S(\HS)}$.

For showing that the inverse $s^{-1}$ of $s : \HL \to \Ic_\leq^{S(\HS)}$ is 
computable we first discuss appropriate admissible representations of $\HL$
and $\HL_q$.\footnote{The coding of $\HL$ is a restriction of a coding of
closed convex subsets of $B(\HL^\prime)$ as can be found in \cite{NeumannE2015}.} 
First notice that $[0,1[ \cap \Rat$ and algebraic numbers can be coded effectively 
by natural numbers. We call elements of $S(\HS)$ ``rational'' iff all items are 
algebraic complex numbers and almost all of them vanish. Thus rational elements 
of $S(\HS)$ can be coded effectively by natural numbers. Elements $d \in \HL_q$ 
are coded by sequences which enumerate all (codes of) pairs $(c,r)$ of rational 
elements $c \in S(\HS)$ and $r \in [0,1[$ such that $r < d(c)$. Elements $L \in \HL$ 
are coded by sequences which enumerate all (codes of) pairs $(c,r)$ of rational 
elements $c \in S(\HS)$ and $r \in [0,1[$ such that $B(\HS) \cap L \subseteq H_{c,r}$.\footnote{We can reconstruct $L$ from $B(\HS) \cap L$ since any code of an element $x$ of 
$\HS^\prime$ different from $0$ can be transformed effectively into a code of an
element $x' \in B(\HS^\prime)$ which is different from $0$ and a multiple of $x$
and thus $x \in L$ iff $x' \in B(\HS) \cap L$.}
By Theorem~\ref{ENthm} an element of Baire space codes $L \in \HL$ iff it codes
$s(L) \in \HL_q$. Thus both $s^{-1}$ and $s$ are coded by any code for the identity
map on Baire space. Since there are effective codes for the latter both $s$ and
$s^{-1}$ are computable.
\end{proof}

There arises the question to which extent the operations on $\HL$ usually considered
in the ``logico-algebraic'' approach do live within $\AdmRep$. Of course, the 
antitonic operation $(-)^\perp : \HL \to \HL$ of orthocomplementation is not 
continuously realizable since otherwise it would be monotonic which is impossible 
since the specialization order on $\HL$ is not discrete. The operation 
$\wedge : \HL \times \HL \to \HL : (P,Q) \mapsto P \cap Q$ is realizable since
the binary supremum operation on $\Sigma$ is effectively realizable.
The following proposition tells us that the binary supremum operation on $\HL$
is not continuously realizable.

\begin{prop}\label{veeHLnotcont}
The function $\vee : \HL \times \HL \to \HL$ where $P \vee Q$ is the least
closed subspace of $\HS$ containing $P$ and $Q$ as subsets is not continuous
and thus not a morphism in $\AdmRep$.
\end{prop} 

\begin{proof}
Suppose $\vee : \HL \times \HL \to \HL$ is continuous. Then it preserves
suprema of ascending $\omega$-chains in each argument. Since $(-)^\perp$
is an anti-automorphism this is equivalent to $P \wedge \bigvee Q_n =
\bigvee P \wedge Q_n$ for all sequence $(Q_n)$ in $\HL$ with 
$Q_n \subseteq Q_{n+1}$.

The following counterexample, however, shows that this is not generally 
the case. Let $Q_n$ be the closed linear subspace of $\HS$ spanned by 
$e_0,\dots,e_n$ and $P$ a $1$-dimensional subspace of $\HS$ with 
$P \cap Q_n = 0$ for all $n$. Since $\bigvee Q_n = \HS$ we have 
$P = P \wedge \bigvee Q_n$ although $\bigvee P \wedge Q_n = 0$.
\end{proof}

Moroever, biorthogonalisation is not continuous as a map from
$\closed(\HS^\prime) \to \HL$.
 
\begin{prop}\label{biorthnotcont}
The biorthogonalization map $(-)^{\perp\perp} : \closed(\HS^\prime) \to \HL$ sending
$C \in \closed(\HS^\prime)$ to $C^{\perp\perp}$ is not continuous and thus not 
a morphism in $\AdmRep$.
\end{prop}

\begin{proof}
Since the topology on $\HL^\prime$ is the sequentialization of the topology 
induced by the inclusion of $\HL$ into $\closed(\HS^\prime)$ the map 
$(-)^{\perp\perp} : \closed(\HS^\prime) \to \HL$ is a continuous if and only if 
$(-)^{\perp\perp} : \closed(\HS^\prime) \to \closed(\HS^\prime)$ is Scott continuous 
which, however, is not the case as the following counterexample shows.

Let $(C_n)$ be a sequence in $\closed(\HS^\prime)$ with $C_{n+1} \subseteq C_n$ such that 
$\bigcap C_n = \{0\}$ and all $C_n^{\perp\perp}$ are the same 1-dimensional subspace 
of $\HS^\prime$. For example one may take $C_n = \{ mx \mid m \geq n \}$ where $x$ is 
some unit vector in $\HS$. Then $(\bigcap C_n)^{\perp\perp} = 0$ although 
$C_n^{\perp\perp} = \CC x$ for all $n$ and thus
$\bigcap C_n^{\perp\perp} = \CC x \neq 0 = (\bigcap C_n)^{\perp\perp}$ 
providing the desired counterexample to (Scott) continuity of biorthogonalization.
\end{proof}

\subsection{States within \texorpdfstring{$\AdmRep$}{AdmRep}}

We want to identify states as particular maps in $\AdmRep$ from the Hilbert lattice 
$\HL$ to the unit interval. Since such maps have to be Scott continuous w.r.t.\
the specialization order we must endow the unit interval with the \emph{upper} 
topology, i.e.\ consider states as maps from $\HL$ to $\Ic$. In the subsequent 
Theorem~\ref{StaAR} we will characterize states as those $s \in \Ic^\HL$ 
which validate the conditions (S1) and (S3) of Prop.~\ref{statchar}. 
 
But for this purpose we need the following result about pure states.

\begin{lem}\label{purestatelm}
For all unit vectors $x \in \HS$ the pure state 
$s_x : \HL \to \Ic : p \mapsto \scpr{x}{P_px}$ is a morphism in $\AdmRep$,
i.e.\ continuous w.r.t.\ the induced topologies.
\end{lem}
  
\begin{proof} 
The claim is an immediate consequence of Th.~\ref{HLtowPr}.
\end{proof}

\begin{thm}\label{StaAR}
A function $s : \HL \to \Ib$ is a state iff $s \in \Ic^\HL$ and it satisfies 
the conditions
\begin{enumerate}
\item[(S1)] $s(0) = 0$ and $s(\HS) = 1$
\item[(S3)] $s(P \vee Q) = s(P) + s(Q)$ whenever $P \perp Q$
\end{enumerate}
of Prop.~\ref{statchar}.
\end{thm}

\begin{proof}
By Prop.~\ref{statchar} a function $s : \HL \to \Ib$ is a state iff it satisfies 
conditions (S1), (S2) and (S3). Since $s \in \Ic^\HL$ preserves suprema of 
$\omega$-chains it always validates (S2) of Prop.~\ref{statchar}. Thus, we have 
shown the backward direction of our claim.

By Lemma~\ref{purestatelm} every pure state is an element of $\Ic^\HL$ and thus
validates condition (S2). Obviously, it always validates conditions (S1) and (S3). 
By Proposition~\ref{convstat} every state arises as a countable convex combination 
of pure states. Any such countable convex combination satisfies conditions (S1) and (S3) 
and is moreover continuous. 
Thus, all states are elements of $\Ic^\HL$ and validate conditions (S1) and (S3).
\end{proof}

Notice that at first sight condition (S3) cannot be expressed in the internal
language since by Prop.~\ref{veeHLnotcont} the operation $\vee$ on $\HL$ is not
a morphism of $\AdmRep$. But under the assumption that $P \perp Q$ a unit vector $x$
is in $P \vee Q$ iff $s(P)(x) + s(Q)(x) = 1$ and thus $x \not\in P \vee Q$ iff
 $s(P)(x) + s(Q)(x) < 1$ in $\Ic$ which is a proposition in $\Sigma$. 
Thus, for $P \perp Q$ we may define $P \vee Q \in \HL$ as 
$\lambda x{:}\HS.\, \|x\| > 0 \wedge 
      s(P)\bigl(\frac{x}{\|x\|}\bigr) + s(Q)\bigl(\frac{x}{\|x\|}\bigr) < 1$
in the internal language. In the light of these considerations the following 
definition makes sense in the internal language of $\RT(\Kc_2)$ and $\KV$.

\begin{defi}\label{StaDef}
Let $\Sta$ be the $\neg\neg$-closed subobject of $\Ic^\HL$ consisting of 
those $s \in \Ic^\HL$ which validate the conditions 
\begin{enumerate}
  \item[(S1)] $s(0) = 0$ and $s(\HS) = 1$
  \item[(S3)] $s(P \vee Q) = s(P) + s(Q)$ whenever $P \perp Q$.
\end{enumerate} 
\end{defi}

Notice that $\Sta$ is a topological predomain since there are no nontrivial 
ascending chains in $\Sta$. For this reason it also lacks a least element and 
thus is not a topological domain. 

\subsection{Observables within in \texorpdfstring{$\AdmRep$}{AdmRep}}\label{BObsrep}

In Proposition~\ref{obschar} we have characterized (quantum) observables as
(Scott) continuous maps $\nu : \closed(\RR) \to \Ic$ such that for every
unit vector $x \in \HS$ the map 
$\nu_x : \closed(\RR) \to \Ic : C \mapsto \scpr{x}{\nu(C)(x)}$
is a probability valuation on $\RR$. That $\nu$ is a family of probability
valuations is axiomatized by the conditions (O1) and (O3) of Proposition~\ref{obschar}.
As follows from the subsequent Lemma~\ref{obscontlem} condition (O2) is equivalent
to $\nu \in \HL^{\closed(\RR)}$.

\begin{lem}\label{obscontlem}
A map $o : \closed(\RR) \to \HL$ is continuous iff it is Scott continuous as a
map from $\closed(\RR)$ to $\closed(\HS^\prime)$.
\end{lem} 
\begin{proof}
Recall that $\closed(\RR) \cong \Sigma^\RR$ in $\AdmRep$ carries the Scott topology
which, moreover, is sequential. Further recall that the topology of $\HL$ is the 
sequentialization of the topology induced by the inclusion of $\HL$ into 
$\closed(\HS^\prime) \cong \Sigma^{\HS^\prime}$. 

The claim follows from the fact that sequentialization is a coreflection.
\end{proof}

Thus, observables are the (global) elements of the object $\Obs$ in
$\AdmRep$ as introduced in the next definition.

\begin{defi}\label{ObsDef}
Let $\Obs$ be the $\neg\neg$-subobject of $\HL^{\closed(\RR)}$ consisting of those 
$\nu \in \HL^{\closed(\RR)}$ such that $\nu_x = \lambda C \in\closed(\RR). s(\nu(C))x$
is a probability valuation for all unit vectors $x \in \HS$ where $s$ is the
continuous map of Theorem~\ref{HLtowPr}. 
\end{defi}

Notice that $\Obs$ is a topological predomain since there are no nontrivial 
ascending chains in $\Obs$. For this reason it also lacks a least element and 
thus is not a topological domain.

\section{A Spectral Theorem for Bounded Observables}

Von Neumann's Spectral Theorem establishes a 1-1-correspondence between 
observables $o$ and unbounded self adjoint operators $A$ on $\HS$ which 
are defined on a dense subspace of $\HS$ on which they are continuous in the 
sense that their graph is a closed subspace of $\HS{\times}\HS$, 
see e.g.\ \cite{Prugovecki1971:QuantumMechanicsInHilbertSpace}.
The correspondence is given explicitly by 
\[ \scpr{x}{Ax} = \int \lambda \, d o_x(\lambda) \]
for $x \in \HS$ where $o_x$ is the measure on $\RR$ given by
$o_x(B) = \scpr{x}{o(B)x}$. Notice that $Ax$ is defined iff the integral
$\int \lambda \, d o_x(\lambda)$ exists.
This restricts to a 1-1-correspondence between bounded self adjoint operators 
$A$ on $\HS$ and observables $o$ which are {\bf bounded} in the sense that 
$o([-c,c]) = \HS$ for some $c \geq 0$.

This correspondence extends to observables formulated in terms of valuations 
since by \cite{CoquandSpitters2009:IntegralsandValuations} integration w.r.t.\ 
valuations can be developed within a fairly weak constructive theory
certainly validated by $\RT(\Kc_2)$ and $\KV$.

We will show now that at least for bounded self adjoint operators this process
can be inverted within $\RT(\Kc_2)$ and $\KV$. Clearly, by rescaling, it suffices to
show this for self adjoint operators bounded by $1$. 

Given a self adjoint operator $A$ bounded by $1$ we consider the commutative 
subalgebra $\mathfrak{A}(A)$ of $\Bfrak(\HS)$ generated by $A$. By the 
Gel'fand-Naimark theorem for commutative $\CC^*$-algebras, which can be proven
constructively and thus holds in $\RT(\Kc_2)$ and $\KV$, the algebra $\mathfrak{A}(A)$ 
is isomorphic to $C(\mathrm{Sp}(\mathfrak{A}(A)))$. Thus, to every unit vector
$x \in \HS$ we may associate the map 
$I_A(x) : C([-1,1]) \to \RR : f \mapsto \scpr{x}{f(A)x}$ 
which is the Daniell-Stone integral corresponding to $\nu_A(x)$ 
where $\nu_A$ is the observable corresponding to $A$ by the Spectral Theorem. 
By the Portmanteau Theorem (Theorem 2.1 of 
\cite{Billingsley99:ConvergenceOfProbabilityMeasures}) 
these Daniell-Stone integrals with the weak topology (of pointwise convergence)
correspond to probability valuations on $[-1,1]$ with the weak topology 
(of pointwise convergence), which as  shown by M.~Schr\"oder in 
\cite{Schroeder2006:AdmissibleRepresentationsOfProbabilityMeasures} coincides 
with the natural $\QCB_0$ topology on probability valuations on $[-1,1]$.
Externally, the observable $\nu_A$ corresponding to $A$ is given by
\[ \nu_A(x)(C) = \inf \{ I_A(x)(f) \mid \chi_C \leq f \in \RR^{[-1,1]}\} \] 
for $x\in\HS$ and $C \in \closed([-1,1])$.
That this correspondence is a homeomorphism follows from the
above mentioned Portmanteau Theorem and the fact that the natural topology on
Daniell-Stone integrals on $[-1,1]$ considered as $\neg\neg$-subobject of
$\Real^{\Real^{[-1,1]}}$ is the topology of pointwise convergence as follows 
from Theorem~\ref{convLO} since the natural topology on $\Real^{[-1,1]}$ coincides
with the metric one induced by the supremum norm.

It remains to show that we can get $\nu_A(x)$ from $I_A(x)$ in an effectively
continuous way. For this purpose we show that $\nu_A$ can be defined from $I_A$
in the internal language of $\RT(\Kc_2)$ and $\KV$ using an argument provided by 
D.~Le\v{s}nik from Univ.~of Ljubljana. 
First notice that we can also define $\nu_A$ as
\[ \nu_A(x)(C) = \inf \{ I_A(x)(f) \mid f \in \RR^{[-1,1]} \mbox{ and } 
                                        \forall x \in C.\, 1 < f(x) \} \]
and, accordingly, also as 
\[ \nu_A(x)(C) = \inf \{ q \in \mathbb{Q} \mid 
\exists f \in \RR^{[-1,1]}.\, I_A(x)(f) < q \wedge
                                        \forall x \in C.\, 1 < f(x) \} \]
Since closed subsets of $[-1,1]$ are compact we have 
$\forall x \in C.\, 1 < f(x) \in \Sigma$ for all $f \in \RR^{[-1,1]}$. 
The proposition $I_A(x)(f) < q$ is in $\Sigma$ anyway.
Let $D$ be some countable dense subset of $\RR^{[-1,1]}$, 
e.g.\ piecewise linear continuous functions with rational ``breakpoints''.
A $\Sigma$-subset of $\RR^{[-1,1]}$ is non-empty iff it has non-empty 
intersection with $D$. Thus, we have
\[ \nu_A(x)(C) = \inf \{ q \in \mathbb{Q} \mid 
\exists f \in D.\, I_A(x)(f) < q \wedge
                                        \forall x \in C.\, 1 < f(x) \} \]
which is an element of $\Ic$ since it arises as infimum of a $\Sigma$-subset
of $\mathbb{Q}$ (because $\Sigma$ is closed under existential quantification
over the countable set $D$). 
  
Thus, we have shown that

\begin{thm}\label{SpecThEff}
The Spectral Theorem for self adjoint operators bounded by $1$ holds 
in $\RT(\Kc_2)$  and $\KV$ and so does - by rescaling - also the 
Spectral Theorem for arbitrary bounded self adjoint operators.
\end{thm}

By Theorem~\ref{convLO} the natural $\QCB_0$ topology on the operator side is the
sequentialization of the strong operator topology. On the side of observables
a sequence $(\nu_n)$ converges to $\nu_\infty$ iff for all unit vectors $x \in \HS$
and $C \in \closed([-1,1])$ the sequence $\nu_n(C)(x)$ converges to 
$\nu_\infty(C)(x)$ in $\Ic$, i.e.\ $\limsup \nu_n(C)(x) \leq  \nu_\infty(C)(x)$.

\section{Conclusion and Future Work}

We have constructed an admissible representation for the Hilbert lattice 
$\HL$ and based on this admissible representations $\Sta$ and $\Obs$ for the 
sets of quantum states and quantum observables, respectively. Thus, these data 
types come endowed with a notion of computability in the sense of Weihrauch's 
Type Two Effectivity where a map between admissible representations is 
computable iff it is realized by an element of $\Kc_{2,\mathit{eff}}$, i.e.\ a 
computable element of Baire space \emph{aka} a total recursive function. 
The corresponding category $\AdmRep_{\mathit{eff}}$ of admissible representations
and effective, i.e.\ computable, maps between them arises as a full reflective
subcategory of the so-called \emph{Kleene-Vesley} topos $\KV$ which is the effective part 
of $\RT(\Kc_2)$ as described in \cite{Oosten2008:Realizability}.

Since $\RT(\Kc_2)$ and $\KV$ are models of Brouwerian intuitionistic mathematics, 
see \cite{KleeneVesley1965:TheFoundationsOfIntuitionisticMathematics,Oosten2008:Realizability}, it appears as natural to develop basic quantum mechanics \emph{synthetically} by 
identifying appropriate axioms holding in the internal language of the respective 
toposes.

We have arrived at our admissible representations of $\HL$, $\Sta$ and $\Obs$
by the fairly abstract methods of topological domain theory. We have used our
abstract account already for proving some basic negative results, namely that 
$(-)^{\perp\perp} : \closed(\HS^\prime) \to \HL$ and $\vee : \HL \times \HL \to \HL$ are not 
continuous and thus not computable. On the positive side in Theorem~\ref{SpecThEff} we
have mangaged to show that the Spectral Theorem for bounded observables does hold
in $\RT(\Kc_2)$ and $\KV$ and thus is continuously and also effectively realizable. 
We conjecture that it can be extended to the general case since the Spectral Theorem 
for unitary operators follows constructively from the one for bounded self adjoint 
operators and from this the general spectral theorem can be obtained using the so-called 
Cayley transform. We expect that these results also hold in $\KV$, i.e.\ have not 
only continuous but also computable realizers.

It might be interesting to make the implicit representations more explicit 
as a basis for actual computation. A particularly challenging question is to which
extent the results of \cite{WeihrauchZhong2006:ComputingSchroedingerPropagatorsOnType2TuringMaschines} can be lifted to our more abstract approach.

\section*{Acknowledgements}

We are grateful to  A.~Bauer, V.~Brattka, K.~Keimel, D.~Le\v{s}nik, 
M.~Schr\"oder, A.~Simpson and B.~Spitters for discussions and helpful hints. 
In particular, we want to thank M.~Schr\"oder for pointing us out that a preliminary
version of Theorem~\ref{HLtowPr} was wrong (see Appendix~\ref{MScounter}) and 
suggesting that we better define $\HL$ as a $\neg\neg$-subobject of $\Sigma^{\HS^\prime}$ 
rather than $\Sigma^\HS$.

\appendix
\section{Why \texorpdfstring{$\HL$}{L} should not be considered as a subobject of \texorpdfstring{$\Sigma^\HS$}{Sigma\^{}H}}\label{MScounter}

In previous versions of this paper we considered the Hilbert lattice $\HL$ 
as the $\neg\neg$-subobject of $\Sigma^\HS$ consisting of all $p$ for which 
$p^{-1}(\bot)$ is a subspace of $\HS$. However, as pointed out to us by 
Matthias Schr\"oder this at first sight compelling definition of $\HL$ does 
not validate the crucial Theorem~\ref{HLtowPr}.
We now describe the counterexample of Schr\"oder.

Let $A_n$ be the closed linear subspace of $\HS$ spanned by the vector 
$e_0 + e_n$. We first show that the sequence $(A_n)$ converges to $\{0\}$. 

Suppose $(x_n)$ is a sequence in $\HS$ converging to $x_\infty$ in $\HS$
such that $\forall n \exists k \geq n \, x_n \in A_n$. Then there exists
a subsequence $(x_{n_k})$ of $(x_n)$ with $x_{n_k} \in A_{n_k}$ for all $k$. 
Then for all $m \geq 1$ it holds that 
$\lim\limits_{k \to\infty} \scpr{e_m}{x_{n_k}} = 0$ and 
thus $\scpr{e_m}{x_\infty} = 0$. 
Suppose $\varepsilon > 0$. Then there exists $n_0$ such that
for all $k \geq n_0$ it holds that $\|x_k - x_\infty\| < \varepsilon$ and thus 
$\scpr{e_m}{x_k} < \varepsilon$ for all $m \geq 1$.
But since $x_k \in A_k$ we have $\scpr{e_0}{x_k} = \scpr{e_k}{x_k}$ and thus 
$\scpr{e_0}{x_k} < \varepsilon$ for $k > n_0$. Thus, we have shown that 
$0 = \lim\limits_{k\to\infty} \scpr{e_0}{x_k} = \scpr{e_0}{x_\infty}$. Accordingly,
we have $x_\infty = 0$ as desired.

For all $n$ we have $s(A_n)(e_0) \geq \frac{1}{2}$, Thus, since 
$s(A_\infty)(e_0) = 0$ it cannot hold that 
$\lim\limits_{n\to\infty} s(A_n)(e_0) = s(A_\infty)(e_0)$, i.e.\ we 
have shown that $s$ is not continuous.

\section{Some Facts about Closed Balls and Half Spaces}\label{ENfacts}

For $c \in \HS$ and $0 \leq r < 1$ we may consider the closed ball
$\overline{B(c,r)} = \{ x \in \HS \mid \| x-c \|\ \leq r \}$ 
and the closed half space 
$h_{c,r} = \{ x \in \HS \mid \Re(\scpr{c}{x}) \geq \sqrt{1 - r^2} \}$.
We write $H_{c,r}$ for the open complement of $h_{c,r}$. 

\begin{lem}\label{ENlm1}
Let $c \in S(\HS)$, $0 \leq r < 1$ and $x \in S(\HS)$. Then 
$|\Re(\scpr{c}{x})| \geq \sqrt{1-r^2}$ iff 
$\lambda x \in \overline{B(c,r)}$ for some  $\lambda\in[-1,1]$.
\end{lem}
\begin{proof}
We have $\lambda x \in \overline{B(c,r)}$ iff $\|c-\lambda x\| \leq r$
iff $\scpr{c-\lambda x}{c-\lambda x} \leq r^2$ iff
\[ 0 \geq \scpr{c-\lambda x}{c-\lambda x} - r^2 = 
   \lambda^2 - 2\lambda\Re(\scpr{c}{x}) + 1 - r^2  =: f(\lambda) \]
Since $r^2 < 1$ we have $f(0) > 0$. Since $f$ is continuous by the
intermediate value theorem we have $f(\lambda) \leq 0$ for some real
$\lambda$ with $|\lambda| \leq 1$ iff $f(\lambda) = 0$ for some real
$\lambda$ with $|\lambda| \leq 1$. 

Since 
$f(\lambda) = (\lambda - \Re(\scpr{c}{x})^2 - \Re(\scpr{c}{x})^2 + 1 - r^2$
the function $f$ has a real zero iff $\Re(\scpr{c}{x})^2 \geq 1 - r^2$, namely
\[ \lambda = \Re(\scpr{c}{x}) - \pm\sqrt{\Re(\scpr{c}{x})^2 - (1-r^2)} \]
which always can be chosen to lie in $[-1,1]$ 
since $|\Re(\scpr{c}{x}| \leq |\scpr{c}{x}| \leq 1$ and 
$|\Re(\scpr{c}{x}| \geq \sqrt{\Re(\scpr{c}{x})^2 - (1-r^2)}$.
\end{proof}

\begin{lem}\label{ENlm2}
Let $c \in S(\HS)$ and $0 \leq r < 1$. If $L$ is a closed linear subspace 
of $\HS$ with $B(\HS) \cap L \cap \overline{B(c,r)} = \emptyset$ then 
$B(\HS) \cap L \cap h_{c,r} = \emptyset$.
\end{lem}
\begin{proof}
We argue by contraposition. So suppose $L$ is a closed linear subspace
of $\HS$ and $B(\HS) \cap L \cap h_{c,r} \neq \emptyset$ for some 
$c \in S(\HS)$ and $0 \leq r < 1$. 

Choose some $y \in B(\HS) \cap L \cap h_{c,r}$. 
Since $y \in h_{c,r}$ and $\sqrt{1 - r^2} > 0$ it follows that $y \neq 0$. 
Thus there exists a unique $\mu \geq 1$ with $x = \mu y$ and $x \in S(\HS)$. 
Thus we have $x \in S(\HS) \cap L \cap h_{c,r}$ from which it follows 
by Lemma~\ref{ENlm1} that for some $\lambda \in [-1,1]$ we have 
$\lambda x \in \overline{B(c,r)}$. 
But we have also $\lambda x \in B(\HS) \cap L$ since $x \in S(\HS) \cap L$.

Thus $B(\HS) \cap L \cap \overline{B(c,r)} \neq \emptyset$ as desired. 
\end{proof}

\begin{lem}\label{ENlm3}
Let $c \in S(\HS)$ and $0 \leq r < 1$. If $L$ be a closed linear 
subspace of $\HS$ with $B(\HS) \cap L \cap h_{c,r} = \emptyset$ 
then $B(\HS) \cap L \cap \overline{B(c,r)} = \emptyset$.
\end{lem}
\begin{proof}
We proceed by contraposition. So suppose 
$B(\HS) \cap L \cap \overline{B(c,r)} \neq \emptyset$.

Then there exists a $y \in B(\HS) \cap L \cap \overline{B(c,r)}$.
Since $y \in \overline{B(c,r)}$ we have $y \neq 0$. 
Thus there exists $x \in S(\HS)$ and $\lambda \in [-1,1]$ with 
$y = \lambda x$. Notice that $\lambda \neq 0$ since $y \neq 0$. 
Thus $x = \frac{1}{\lambda}y \in S(\HS) \cap L$ from which it follows 
by Lemma~\ref{ENlm1} that $|\Re(\scpr{c}{x})| \geq \sqrt{1-r^2}$.
If $\Re(\scpr{c}{x} \geq 0$ then $x \in h_{c,r}$. Otherwise we have 
$-x \in h_{c,r}$ and thus since $B(\HS) \cap L$ is closed under 
multiplication with $-1$ we also have $-x \in B(\HS) \cap L$.
Thus in any case $h_{c,r}$ and $B(\HS) \cap L$ are not disjoint 
as desired.
\end{proof}

\begin{lem}\label{ENlm4}
Let $L$ be a closed linear subspace of $\HS$ and $c \in S(\HS)$. Then we have
\[ B(\HS) \cap L \cap h_{c,r} = \emptyset \quad\mbox{iff}\quad 
   L \cap \overline{B(c,r)} = \emptyset  \quad\mbox{iff}\quad 
   d(c,L) > r \]
for all $0 \leq r < 1$.
\end{lem}
\begin{proof}
Let $c \in S(\HS)$, $0 \leq r < 1$ and $L$ a closed linear subspace  
of $\HS$. The second equivalence is almost tautological.
Thus we concentrate on the first equivalence.

From Lemma~\ref{ENlm2} and lemma~\ref{ENlm3} it follows that
\[ B(\HS) \cap L \cap h_{c,r} = \emptyset \quad\mbox{iff}\quad 
   B(\HS) \cap L \cap \overline{B(c,r)} = \emptyset \]
and thus for establishing our claim it remains to show that
\[  L \cap \overline{B(c,r)} = \emptyset \quad\mbox{iff}\quad 
    B(\HS) \cap L \cap \overline{B(c,r)} = \emptyset \]
which, however, follows from the following consideration.

Since the forward direction is trivial it suffices to prove the
contraposition of the backwards direction. For this purpose
suppose $L \cap \overline{B(c,r)} \neq \emptyset$, i.e.\ $d(c,L) \leq r$.
Thus $d(c,P_L(c)) = d(c,L) \leq r$, i.e.\ $P_L(c) \in \overline{B(c,r)}$.
Since $P_L(c) \in B(\HS)$ it follows that 
$B(\HS) \cap L \cap \overline{B(c,r)} \neq \emptyset$ as desired.
\end{proof}

\begin{thm}\label{ENthm}
Let $L$ be a closed linear subspace of $\HS$ and $c \in S(\HS)$. Then we have
\[ B(\HS) \cap L \subseteq H_{c,r}  \quad\mbox{iff}\quad   d(c,L) > r \]
for all $0 \leq r < 1$.
\end{thm}
\begin{proof}
Immediate from Lemma~\ref{ENlm4} by observing that for $0 \leq r < 1$ 
we have
\[ B(\HS) \cap L \subseteq H_{c,r}  \quad\mbox{iff}\quad 
   B(\HS) \cap L \cap h_{c,r} = \emptyset \]
since by definition $H_{c,r}$ is the complement of $h_{c,r}$.
\end{proof}

\end{document}